\documentclass{article}

\usepackage{amssymb, dsfont}
 \usepackage{amsthm}
 \newtheorem{theorem}{Theorem}[section]
\newtheorem{lemma}[theorem]{Lemma}

\theoremstyle{definition}

\usepackage{mathtools}
\usepackage{xr}
\usepackage{lineno}
\usepackage{xcolor}
\usepackage{algorithm} 
\usepackage{algpseudocode}
\usepackage{pgfplots}
\usepackage{authblk}
\usepackage{natbib}

\title{Multi Split Conformal Prediction}
\author[1]{Aldo Solari}
\author[2]{Vera Djordjilovi\'c}
\affil[1]{Department of Economics, Management and Statistics, University of Milano-Bicocca}
\affil[2]{Department of Economics, University Ca' Foscari of Venice}
\date{\today}

\begin{document}

\maketitle

\begin{abstract}
Split conformal prediction is a computationally efficient method for performing distribution-free predictive inference in regression. 
It involves, however, a one-time random split of the data, and the result can strongly depend on the particular split. 
To address this problem, we propose multi split conformal prediction, a simple method based on Markov’s inequality to aggregate split conformal prediction intervals across multiple splits.
\end{abstract}




\section{Introduction}
\label{sec:1}
Conformal prediction is a general framework for constructing prediction sets in regression problems with finite sample and distribution free marginal coverage  \citep{vovk2005algorithmic, shafer2008tutorial, gammerman2013learning}. Under very mild assumptions, conformal predictions sets provide exact coverage. These appealing theoretical properties are contrasted by very high computational cost, which hinders its practical application.

To address this issue, \cite{papadopoulos2002inductive} and \cite{lei2018distribution} have proposed inductive or \emph{split conformal prediction} which successfully  addresses the issue of computational efficiency, but at the cost of introducing extra randomness due to a one-time random split of the data. This kind of randomness of the prediction interval  parallels the  ``$p$-value lottery'' discussed in \cite{meinshausen2009p}. 


 
A straightforward strategy for alleviating this issue is  to  aggregate results from  different data splits \citep{carlsson2014aggregated}, which is the basis for cross-conformal prediction  \citep{vovk2015cross},  jackknife+  and $K$-fold CV+ prediction \citep{barber2021predictive}, as well as $K$-subsample conformal prediction of \cite{gupta2019nested}.    Table \ref{tab:aggregated} provides an overview of their finite sample coverage guarantees. 

\begin{table}[ht]\label{tab:aggregated}
        \centering
        \begin{tabular}{|lll|}
        \hline
        \emph{Method} & \emph{Coverage}  & \emph{Reference} \\
        \hline
         Cross-conformal & $\geq 1-2\alpha - a(n,K)$ & \cite{vovk2015cross} \\
         Jackknife+/CV+ & $ \geq 1-2\alpha - \min\{a(n,K), b(n,K)\}$ & \cite{barber2021predictive} \\
         Subsampling conformal & $\geq 1-\min\{2,K\}\alpha$ & \cite{gupta2019nested}\\
                 \hline
        \end{tabular}
        \caption{Aggregated conformal prediction methods with proven coverage guarantees, where $a(n,K)=(2-2/K)/(n/K +1)$ and $b(n,K)=(1-K/n)/(K+1)$.}
\end{table}
    
The coverage guarantees  listed in Table \ref{tab:aggregated} are based on the fact that double of the average $p$-value is a valid $p$-value, a result established by \cite{ruschendorf1982random} and discussed in \cite{vovk2020combining}. Only the factor $b(n,K)$ derived in Theorem 4 of \cite{barber2021predictive}  is based on a different argument that makes use of Landau's theorem for tournaments \citep{landau1953dominance}.  

In this work, we propose \emph{multi split conformal prediction}, a simple method based on Markov's inequality to aggregate split conformal prediction intervals across multiple splits. The proposed method is similar in spirit to $p$-value aggregation \citep{van2009testing, meinshausen2009p, romano2019multiple} and stability selection \citep{meinshausen2010stability, shah2013variable, ren2020derandomizing}. In particular, the multi split prediction set includes those points that are included in single split prediction intervals with frequency greater than a user defined threshold. Notably, the Bonferroni-intersection  method of  \cite{lei2018distribution} and the jackknife+/CV+ of \cite{barber2021predictive} can be seen as special cases of the proposed approach. 

In Sections \ref{sec:2} and \ref{sec:3} we revisit full and split conformal prediction, highlighting a necessary and sufficient condition for obtaining exact coverage. 
The main result regarding the aggregation of single split intervals is presented in Section \ref{sec:4}. 
An illustration of the proposed method is given in Section \ref{sec:6}.

\section{Conformal prediction}
\label{sec:2}

Assume that $Z_i=(X_i,Y_i)$, $i \in [n+1]$ are $n+1$ independent identically distributed random vectors from a probability distribution $P_{XY}$ on the sample space $\mathcal{X} \times \mathcal{Y}=\mathbb{R}^d \times \mathbb{R}$, where $[n]$ denotes $\{1,\ldots,n\}$. 
Suppose that the realizations $z_i= (x_i,y_i)$, $i\in [n]$ and $x_{n+1}$ are available, and we want to predict $Y_{n+1}$ based on $x_{n+1}$. 
We aim to construct a prediction set $C_{\alpha}(x)=C_{\alpha}(x; Z_i, i \in [n]) \subseteq \mathbb{R}$ such that its marginal coverage is at least $1-\alpha$, i.e. 
\begin{eqnarray}\label{coverage}
\mathrm{pr}(Y_{n+1} \in C_{\alpha}(X_{n+1})) \geq 1-\alpha,
\end{eqnarray}
where the probability is taken over all $Z_i$, $i \in [n+1]$.

Let $\phi_\alpha = \phi_{\alpha}(Z)\in \{0,1\}$ be a Bernoulli random variable, where
$Z = (Z_i,i \in [n+1])$. Denote by $\phi^y_\alpha = \phi_{\alpha}(Z^y)$ with $Z^y = (Z_1,\ldots,Z_n, Z^y_{n+1})$ and $Z^y_{n+1}=(X_{n+1},y)$.

\begin{theorem}\label{th:conformal}
Assume that $\phi_\alpha$ is a Bernoulli random variable such that $\mathrm{E}(\phi_\alpha)\leq \alpha$. 
Then the prediction set 
\begin{eqnarray*}\label{C}
C_{\alpha}(x) = \{y \in \mathbb{R}: \phi^y_\alpha = 0\},
\end{eqnarray*}
satisfies (\ref{coverage}). Exact coverage $\mathrm{pr}(Y_{n+1} \in C_{\alpha}(X_{n+1})) = 1-\alpha$ is obtained if and only if $\mathrm{E}(\phi_\alpha) = \alpha$.
\end{theorem}

Informally, $\phi_{\alpha}^y$ can be thought of as a test for the null hypothesis that $Y_{n+1}$ assumes the value of $y$, that is $H_y: Y_{n+1}=y$. Theorem \ref{th:conformal} states that a valid prediction set can be obtained by inverting a collection of such tests. The proof of Theorem \ref{th:conformal} and of all the following results are provided in the Supplementary material.

\section{Split conformal prediction}
\label{sec:3}

Consider a partition of $[n]$ into a calibration set $L$ of size $w$ and a validation set $I$ of size $m = n-w$, independently of the observed data values.  Define a statistic $R= R(Z_L, Z_{n+1})$, referred to as conformity score in conformal inference, to serve as a measure of plausibility of the value $y$ as a realization of $Y_{n+1}$ for the observed value of $X_{n+1}$. Examples include 
\begin{equation}
R = |Y_{n+1}-\hat{\mu}_L(X_{n+1})|,
\label{r1}
\end{equation}
where $\hat{\mu}_L$ is an estimator of $\mathbb{E}(Y\mid X)$ based on $(Z_{l})_{l\in L}$ \citep{papadopoulos2002inductive, lei2018distribution} and
\begin{equation}
\label{r2}
R = \max\left\{\hat{q}_L^\gamma (X_{n+1})- Y_{n+1}, Y_{n+1} - \hat{q}_L^{1-\gamma} (X_{n+1})\right\},
\end{equation}
where $\hat{q}_L^\gamma$ is an estimator of the $\gamma$-quantile of $Y\mid X$ \citep{romano2019conformalized,sesia2020comparison}.
Denote the validation set by $I=\left\{j_1, \ldots, j_{m}\right\}$ and let
\begin{equation}
    R_i = R((Z_l)_{l\in L}, Z_{j_i}), \quad i\in [m].
\end{equation}
For $\alpha \in (0,1)$, define a quantile $R_{\alpha}= R_{\lceil (1-\alpha)(m+1)\rceil}$, where $R_{1}\leq\ldots\leq R_{m}$ are ordered statistics obtained by sorting $R_1, \ldots, R_m$ in  non-decreasing order with ties  broken arbitrarily.

\begin{lemma}\label{le:phi}
The Bernoulli variable
$\phi_\alpha = \mathds{1}\{R > R_{\alpha}\}$ satisfies $\mathrm{E}(\phi_\alpha)\leq \alpha$. If $R_1,\ldots,R_m,R$ are almost surely distinct,
then $\mathrm{E}(\phi_\alpha) = \alpha$  if and only if $\alpha \in  \{1/(m+1),2/(m+1),\ldots,m/(m+1)\}$.
\end{lemma}


Algorithm \ref{alg:SC} describes how to compute the split conformal prediction set. 

\begin{algorithm}[ht]
\caption{Split Conformal}\label{alg:SC}
 	\begin{algorithmic}[1]
 	\Require{data $(x_1,y_1),\ldots,(x_{n},y_{n})$, $x_{n+1}$, validation sample size $m$,  statistic $R$, level $\alpha \in (0,1)$}
     \State split $[n]$ into  $L$ of size $w$ and  $I$ of size $m=n-w$
     \State compute $\{R_{i}\}_{i=1}^{m}$ and $R_{\alpha} = R_{(\lceil (1-\alpha)(m+1)\rceil)}$ 
 	\end{algorithmic} 
 	\Return  split conformal prediction set $C_{\alpha}(x_{n+1}) = \{y \in \mathbb{R}: R \leq R_{\alpha}\}$
\end{algorithm}
 In particular, for $R$ defined as in \eqref{r1} and \eqref{r2}, 
Algorithm \ref{alg:SC}  returns  $C_{\alpha}(x_{n+1}) = [\hat{\mu}_{L}(x_{n+1}) - R_{\alpha},\hat{\mu}_{L}(x_{n+1}) + R_{\alpha}]$ and $C_{\alpha}(x_{n+1}) = [\hat{q}_L^{\gamma}(x_{n+1}) - R_\alpha, \hat{q}_{L}^{1-\gamma}(x_{n+1})+R_{\alpha}]$, respectively. The former is always an interval, whereas the latter is either an  interval or an empty set, i.e. $C_{\alpha}(x_{n+1}) = \emptyset$ if and only if $R_{\alpha} < (1/2)[\hat{q}_{L}^{\gamma}(x_{n+1}) - \hat{q}_{L}^{1-\gamma}(x_{n+1})]$ \citep{gupta2019nested}.

\section{Multi split conformal prediction}
\label{sec:4}

The multi split approach consists in constructing single split prediction sets multiple times, and proceeds with aggregating the results by including those points that are included in single split prediction intervals with frequency greater than a threshold. 

We proceed as follows: we first choose the number of splits $B \in \mathbb{N}$. We then partition $[n]$ into $L^{[b]}$ of size $w^{[b]}$ and $I^{[b]}$ of size $m^{[b]} = n-w^{[b]}$, independently of the observed data values, and choose a statistic $R^{[b]}$, for $b=1,\ldots,B$. For $\beta \in (0,1)$, the Bernoulli random variable $\phi_{\beta}^{[b]} = \mathds{1}\{ R^{[b]} > R^{[b]}_{\beta} \}$ has expected value $\mathrm{E}(\phi_{\beta}^{[b]}) \leq \beta$ by Lemma \ref{le:phi}. Let 
\begin{eqnarray}\label{V}
V_{\beta} =\sum_{b=1}^{B} \phi_{\beta}^{[b]}
\end{eqnarray}
be the number of successes (1s). 
The following Theorem provides an upper bound for $\mathrm{pr}( V_{\beta} \geq k)$, the probability of at least $k$ successes out of $B$ trials.

\begin{theorem}\label{th:markov}
Let $\lambda$ be a non-negative integer such that, for a given integer $1\leq k \leq B$ and $\beta \in (0,1)$, the following holds: 
 \begin{eqnarray}\label{as:lambda}
\sum_{u=0}^{k-1}\mathrm{pr}(V_{\beta} \in [k-u,k) ) \geq \sum_{u=0}^{\lambda}\mathrm{pr}(V_{\beta} \in [k,k+u) ).
\end{eqnarray}
Then
\begin{eqnarray}\label{lambdaBound}
\mathrm{pr}(V_{\beta} \geq k) \leq \frac{B\beta}{k + \lambda}.
\end{eqnarray}
\end{theorem}

\noindent The  parameter $\lambda$ can  be  regarded  as  a  smoothing  parameter. 
The value $\lambda = 0$ reduces (\ref{lambdaBound}) to Markov's bound, while positive values of $\lambda$ correspond to tighter bounds achievable under constraints on the shape of the
distribution of $V_\beta$ \citep{shah2013variable, huber2019halving, ren2020derandomizing}. For $k=1$ and $\lambda = B-1$, assumption (\ref{as:lambda}) holds if and only if  $\phi^{[1]}_\beta=\ldots=\phi^{[B]}_\beta$, and for $k=B$, it holds if
$p_1 + 2 p_2 + \ldots +  (B-1) p_{B-1}
\geq  \lambda p_B$, where $p_k$ denote $\mathrm{P}(V_\beta = k)$. For an odd number $B$ with $k=(B+1)/2$ and $\lambda = k-1$, assumption (\ref{as:lambda}) requires that the probability mass function of $V_{\beta}$ on $\{1,\ldots,B-1\}$ is not skewed to the right in the sense that
 \begin{eqnarray}\label{skewk}
p_1 + 2p_2 + \ldots + \frac{B-1}{2} p_{\frac{B-1}{2}} \geq  \frac{B-1}{2} p_{\frac{B+1}{2}} + \ldots + 2 p_{B-2} + p_{B-1}.
\end{eqnarray}
In general, however, we do not have a guarantee that a positive value of $\lambda$ satisfying \eqref{as:lambda} exists. 

Theorem \ref{th:markov} can be used to aggregate results of split conformal inference performed over a number of different data splits. Let
\begin{eqnarray*}\label{Pi}
\Pi_{\beta}  = 1 - \frac{V_{\beta}}{B} =  \frac{1}{B} \sum_{b=1}^{B} \mathds{1}\{Y_{n+1} \in C^{[b]}_{\beta}(X_{n+1})\}
\end{eqnarray*}
be the proportion of prediction sets $C_{\beta}^{[b]}(X_{n+1})$ that include $Y_{n+1}$. For $\alpha \in (0,1)$ and a threshold $\tau = 1 - k/B$, the multi split conformal prediction set defined as
\begin{eqnarray}\label{MSC}
C^\tau_{\alpha}(x_{n+1})=\{y \in \mathbb{R}: \Pi_{\beta}^y  > \tau \}
\end{eqnarray}
has coverage at least $1-\alpha$ by Theorem \ref{th:conformal} with $\phi_{\alpha} = \mathds{1}\{V_\beta \geq k\} = \mathds{1}\{\Pi_\beta \leq \tau\}$, where $\beta = \alpha(1-\tau)$ with no assumptions or $\beta = \alpha(1-\tau + \lambda/B)$ under the assumption (\ref{as:lambda}) of Theorem \ref{th:markov} guarantees $\mathrm{E}(\phi_\alpha) \leq \alpha$. 

Algorithm \ref{alg:MSC} describes how to compute the multi split conformal prediction set. 

\begin{algorithm}[ht]
\caption{Multi Split Conformal}\label{alg:MSC}
 	\begin{algorithmic}[1]
 	\Require{data $(x_1,y_1),\ldots,(x_{n},y_{n})$, $x_{n+1}$, number of splits $B \in \mathbb{N}$, calibration sample sizes  $(m^{[b]})_{b=1}^{B}$, statistics $(R^{[b]})_{b=1}^{B}$, threshold $\tau \in [0,(B-1)/B]$, level $\alpha \in (0,1)$, smoothing parameter $\lambda\in \mathbb{N}_0$.}
 	\For{$b \gets 1$ to $B$ }  
 \State compute  $C^{[b]}_{\beta}(x_{n+1})$ using Algorithm \ref{alg:SC} with $m^{[b]}$, $R^{[b]}$ and level $\beta =  \alpha(1-\tau + \lambda/B)$ 
      	\EndFor
 	\end{algorithmic} 
 	\Return multi split conformal prediction set $C^\tau_{\alpha}(x_{n+1})=\{y \in \mathbb{R}: \Pi_{\beta}^y  > \tau \}$
\end{algorithm}

In general, $C^\tau_\alpha$ is not guaranteed to be an interval, even when single split prediction sets $C^{[b]}_\beta$ are all intervals. To compute $C^\tau_\alpha$  efficiently,  one can use Algorithm 1 in \cite{gupta2019nested}.
If the computation of each single split interval $C^{[b]}_\beta$ takes time $\leq T$, the overall time to compute $C^\tau_\alpha$ is $O(B \log B) + BT$.

The parameter $\tau$ can be regarded as a tuning parameter, and proper choice of $\tau$ is essential for good performance. Consider the case without assumption, i.e. $\lambda = 0$. On the one hand, setting $\tau=1-1/B$ gives the Bonferroni-intersection method of \cite{lei2018distribution} with $C^{1-1/B}_{\alpha} = \bigcap_{b} C^{[b]}_{\alpha/B}$.
On the other hand, setting $\tau=0$ gives an unadjusted-union $ C^0_\alpha = \bigcup_{b} C^{[b]}_{\alpha}$.

An intermediate choice $\tau=1/2$ amounts to constructing $B$ single split confidence intervals  at level $\alpha/2$, that is $C^{[b]}_{\alpha/2}$, which is a small but not negligible price to pay for using multiple splits rather than just one split. However, if the assumption of Theorem (\ref{th:markov}) holds for $\lambda>0$, we can use a higher level $\alpha(1/2 +  \lambda/B)$, that is  $C^{[b]}_{\alpha(1/2 +  \lambda/B)}$. 

Notice that the multi split prediction set $C^\tau_\alpha$ is very flexible because it allows to use splits of different proportions $m/n$ and possibly different statistics $S$ across splits. This flexibility can be especially useful in conformal quantile regression since it allows to consider different values for the quantile $\gamma$ in \eqref{r2}. The value of $\gamma$ and the proportion $m/n$ seem to be critical for the performance of conformal quantile regression, as highlighted by \cite{romano2019conformalized} and \cite{sesia2020comparison}.

\section{Example}
\label{sec:6}

We apply multi split conformal prediction on the Communities and Crime data set 
analyzed in \cite{romano2019conformalized, sesia2020comparison} and \cite{barber2021predictive}. The data set contains information on 1994 communities, with information regarding median family income, family size, per capita number of police officers, etc. The goal is to predict a response variable defined as the per capita violent crime rate. After removing categorical variables and variables
with missing data, $d = 99$ covariates remain.

\begin{figure}[ht]
    \centering
\includegraphics[scale=.8]{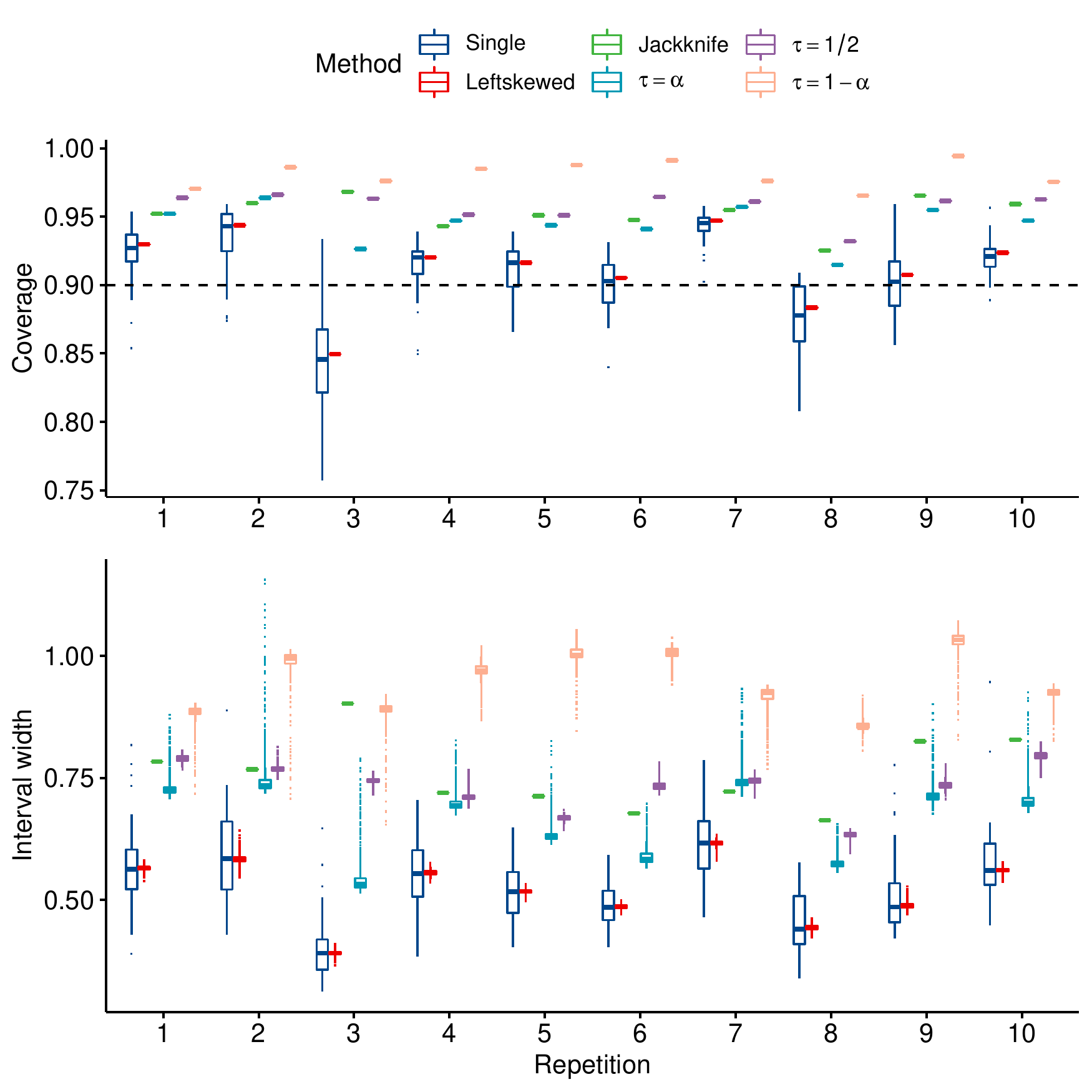}
    \caption{
    Coverage and interval width for single split and multi split prediction sets on the Communities and Crime data set.
     Training and test sets of sizes $n=200$ and $1794$. The coverage level is $1-\alpha=90\%$.
    The statistic $R$ is defined as in (\ref{r1}) where $\hat{\mu}_L$ is estimated by ridge regression. Multi split is performed with $\tau = \alpha, 1/2, 1-\alpha$ and $(\tau,\lambda) = ((B-1)/2B,(B-1)/2))$ (Leftskewed) by assuming (\ref{skewk}). The number of random splits is 
    $B=51$, and the size of the inference set is $m=99$.
     The experiment is repeated 10 times by randomly splitting the train/test each time.}
    \label{fig:Crime}
\end{figure}

We replicate the experiment in \cite{barber2021predictive}. 
We randomly sample $n = 200$ data points from the full data set, to use as training data. The remaining 1794 points form the test set. We use $R$ defined as in \eqref{r1} where $\hat{\mu}_L$ is estimated by the ridge regression algorithm with penalty parameter chosen as $0.001 c^2$, where $c$ is the largest singular value of the training data matrix. 
We set the coverage level to $1-\alpha=90\%$, the number of random split to $B=51$ and size for the inference set to $m=99$. We construct $B$ single split intervals and the multi split interval by using $\tau = \alpha$, $\tau = 1/2$, $\tau=1-\alpha$ with no assumptions and $(\tau, \lambda) = ((B-1)/2B, (B-1)/2)$ by assuming (\ref{skewk}), which we refer to as ``Leftskewed''. The jackknife+ interval is added for comparison with $\alpha=5\%$ in order to guarantee coverage at $1-2\alpha=90\%$. For each method, we calculate its empirical coverage and interval width on the test set. 
We then repeat this procedure
10 times, with the train/test split formed randomly each time.

Figure \ref{fig:Crime} displays the results.
Intervals obtained by the ``Leftskewed'' method exhibit coverage and width comparable to single split intervals, but with substantially reduced variability, as expected. Assumption free methods reflect the conservativeness of Markov's inequality and can not compete with the exact coverage single split method.

 \section{Discussion}
 
 We have proposed a simple method for aggregating single split conformal prediction intervals. In general, the proposed method is conservative and can not compete with the exact single split method. However, it reduces the randomness of a single data split and provides flexibility in combining different statistical learning algorithms across different splits. In addition, we have shown that the conservativeness of the method can be attenuated under an additional  mild assumption that sharpens Markov's tail inequality.  Investigating this and similar assumptions necessitates future research.



\bibliographystyle{elsarticle-harv} 
\bibliography{MSCP_shortjournalname.bib}

\appendix

\section{Proofs}

\begin{proof}[Proof of Theorem \ref{th:conformal}]
$\mathrm{pr}(Y_{n+1} \notin C_{\alpha}(X_{n+1})) =   \mathrm{E}(\phi_\alpha) \leq \alpha$. 
\end{proof}


\begin{proof}[Proof of Lemma \ref{le:phi}]
We provide an explicit formulation of split conformal prediction within the permutation framework. Consider the group of transformations $\Sigma = \{\sigma_1,\ldots,\sigma_{m+1}\}$ whose $m+1$ elements are restricted permutations consisting of swapping the index $n+1$ with another index of $I\cup\{n+1\} = (j_1,\ldots,j_m, n+1)$, i.e. for $i\in [m]$,
$\sigma_i = (\sigma_{i}(j_1),\ldots,\sigma_{i}(j_m), \sigma_{i}(n+1))$ is such that $\sigma_i{(n+1)}=j_i$, $\sigma_{i}(j_i)=n+1$ and $\sigma_i(j_k) = j_k$ for $j_k\neq j_i$. Here $\sigma_{m+1}$ denotes the identity permutation. 
Note that $\Sigma$ is a group with respect to the operation of composition of transformations: $\Sigma$ contains an identity element; every element of $\Sigma$ has an inverse in $\Sigma$; for all $\sigma$, $\tilde{\sigma} \in \Sigma$, $\sigma \circ \tilde{\sigma} \in \Sigma$.

For any $\sigma \in \Sigma$, let $\sigma Z = (Z^*_1,\ldots,Z^*_{n+1})$ be the transformed vector with $Z_i^* = Z_{\sigma(i)}$ if $i \in I \cup \{n+1\}$ and $Z_i^* = Z_i$ otherwise, and let
\begin{eqnarray*}\label{Rperm}
R(\sigma Z)  = R(Z_L, Z_{\sigma(n+1)})
\end{eqnarray*}
be the statistic $R$ calculated on $\sigma Z$. 

Since $Z_1,\ldots,Z_{n+1}$ are independent and identically distributed by assumption,
$Z  \stackrel{d}{=} \sigma Z$ holds for every $\sigma \in \Sigma$. This implies 
the group invariance condition \citep{hoeffding1952large, lehmann2006testing,
hemerik2018exact, hemerik2020another}:
\begin{eqnarray}\label{GIC}
(R(\sigma_1 Z),\ldots, R(\sigma_{m+1} Z)) \stackrel{d}{=} (R(\sigma_1 \circ \sigma Z),\ldots, R(\sigma_{m+1} \circ \sigma Z))
\end{eqnarray}
for every $\sigma \in \Sigma$, where $\stackrel{d}{=}$ denotes equality in distribution. Note that (\ref{GIC}) is implied by exchangeability of  $(R_1,\ldots,R_{m},R)$  \citep{commenges2003transformations, kuchibhotla2020exchangeability}, where $R_1 = R(\sigma_1 Z)$, $\ldots $, $R_m = R(\sigma_{m} Z)$, $R = R(\sigma_{m+1} Z)$.

For $\alpha \in (0,1)$, let
$\tilde{R}_{(1)} \leq \ldots \leq \tilde{R}_{(m+1)}$ be the sorted values of $R_1, \ldots ,R_m, R$, with ties broken arbitrarily, and let 
\begin{eqnarray*}\label{Ralpha}
\tilde{R}_{\alpha} =  \tilde{R}_{(k)}
\end{eqnarray*}
with $k = \lceil (1-\alpha)(m+1)\rceil$, be the the $k$th ordered statistic. We have $\tilde{R}_{\alpha} > R_\alpha$ if $R \leq \tilde{R}_{\alpha}$ and $\tilde{R}_{\alpha} = R_\alpha$ if $R > \tilde{R}_{\alpha}$. Then $\phi_\alpha =  \mathds{1}\{R > \tilde{R}_{\alpha}\}= \mathds{1}\{R > R_{\alpha}\}$.

 From Theorem 1 in \cite{hemerik2018exact}, we obtain that 
the Bernoulli variable
$\phi_\alpha = \mathds{1}\{R > \tilde{R}_{\alpha}\}$ satisfies $\mathrm{E}(\phi_\alpha)\leq \alpha$. 

Finally, Condition 1 and Proposition 1 in \cite{hemerik2018exact} ensure that if $R_1,\ldots,R_m,R$ are almost surely distinct,
then $\mathrm{E}(\phi_\alpha) = \alpha$  if and only if $\alpha \in  \{1/(m+1),2/(m+1),\ldots,m/(m+1)\}$.

\end{proof}

\begin{proof}[Proof of Theorem \ref{th:markov}]

The following proof follows the lines of Lemma 2 of \cite{ren2020derandomizing}.  See also \cite{huber2019halving}.

For each   $\lambda \in \mathbb{N}_0$, we introduce an auxiliary random variable $U \sim \mathrm{Unif}(-\lambda,k)$  independent of $V_\beta$. We will prove the following:
$$\mathrm{pr}(V_\beta \geq k ) \leq \mathrm{pr}(V_\beta + U \geq k )\leq \frac{\mathrm{E}(V_\beta)}{k+ \lambda}.$$
For the first inequality, we have
\begin{eqnarray*}
\mathrm{pr}(V_\beta + U \geq k) 
&=& \frac{1}{\lambda+k}\int_{-\lambda}^{k}\mathrm{pr}(V_\beta\geq k-u)\, \mathrm{d} u \\
&=& \frac{1}{\lambda+k}\left[\int_{-\lambda}^{0}\mathrm{pr}(V_\beta\geq k-u)\, \mathrm{d} u + \int_{0}^{k}\mathrm{pr}(V_\beta\geq k-u)\, \mathrm{d} u\right]\\
&=& \frac{1}{\lambda+k}\left\{\int_{-\lambda}^{k}\mathrm{pr}(V_\beta\geq k)\, \mathrm{d} u +\left[ \int_{0}^{k}\mathrm{pr}\left(V_\beta\in [k-u,k)\right)\, \mathrm{d} u - \int_{-\lambda}^0 \mathrm{pr}\left(V_\beta \in (k,k-u]\right)\mathrm{d}u\right]\right\}\\
&\geq& \mathrm{pr}(V_\beta \geq k).
\end{eqnarray*}
The inequality  follows from  the fact that the term in the squared brackets is non-negative. Namely, since $V_\beta$ is a discrete random variable, we have 
$$
\int_{0}^{k}\mathrm{pr}\left(V_\beta \in [k-u, k)\right)\mathrm{d}u = \sum_{u=1}^{k-1}\mathrm{pr}(V_\beta \in [k-u, k)),
$$
and similarly 
$$
\int_{-\lambda}^0 \mathrm{pr}\left(V_\beta \in (k, k-u]\right) \mathrm{d}u = \sum_{u=1}^\lambda \mathrm{pr}\left(V_\beta \in (k, k+u])\right),
$$
which together with assumption
\eqref{as:lambda} completes the proof of the first inequality. 

\noindent The second inequality follows from 
\begin{eqnarray*}
\mathrm{pr}(V_\beta + U \geq k) &=& 
 \mathrm{E}(\mathds{1}\{k-V_\beta \leq - \lambda\}) + \mathrm{E}\Big(\frac{V_\beta}{ k+\lambda } \mathds{1}\{ -\lambda < k-V_\beta \leq k\}\Big) \\
&\leq & \mathrm{E}\Big(\frac{V_\beta}{ k+\lambda } \Big).
\end{eqnarray*}

\end{proof}

\section{Cross conformal prediction}

The multi split approach includes  $B$-fold cross-conformal prediction as a special case, $2\leq B \leq n$. Suppose for simplicity that $n=Bm$. Consider $(j_1,\ldots,j_n)$, a permutation of $[n]$, and let $L^{[1]}$ be the complement of $I^{[1]}=(j_1,\ldots,j_{m})$, $L^{[2]}$ be the complement of $I^{[2]}=(j_{m+1},\ldots,j_{2m})$, $\ldots$, $L^{[B]}$ be the complement of $I^{[B]}=(j_{(B-1)m + 1},\ldots,j_n)$. Then $B$-fold cross-conformal prediction is obtained by using  Algorithm \ref{alg:MSC} with $L^{[b]},I^{[b]}$ for $b=1,\ldots,B$.

A special case corresponds to $B=n$, resulting in leave-one-out conformal prediction. Let
\begin{eqnarray*}\label{Ri}
R^{[i]} = R((Z_l)_{l\neq i}, Z_{n+1}), \quad R^{[i]}_{1/2} =R((Z_l)_{l \neq i}, Z_{i}), \quad \phi^{[i]}_{1/2} = \mathds{1}\{ R^{[i]} > R^{[i]}_{1/2} \}, \quad i=1,\ldots,n,
\end{eqnarray*}
and $V_{1/2} = \sum_{i=1}^{n} \phi^{[i]}_{1/2}$. 
For $k=\lfloor (1-\alpha)(n+1) \rfloor$, it gives
\begin{eqnarray*}\label{LOO}
C^{\mathrm{L00}}_{2\alpha} = \Big\{ y \in \mathbb{R}: \Pi_{1/2}^y > \frac{\alpha(n+1) - 1}{n} \Big\}
\end{eqnarray*}
where $\Pi_{1/2}^y = \frac{1}{n} \sum_{i=1}^{n} \mathds{1}\{y\in C^{[i]}_{1/2}(x_{n+1})\}$. 

For $R$ defined as in (\ref{r1}) and symmetric in its first argument, \cite{barber2021predictive} proposed an interval that always contains $C^{\mathrm{L00}}_{2\alpha}$, called the jackknife+ prediction interval, and \cite{balasubramanian2014conformal} showed that the coverage  guarantee for $C^{\mathrm{L00}}_{2\alpha}$ is at least $1-2\alpha$. 


\section{Communities  and  Crime  data  set}

We randomly sample $n = 200$ data points from the full data set, to use as training data. The remaining 1794 points form the test set. 
The procedure is repeated
1000 times, with the train/test split formed randomly each time. For each method, we calculate its empirical coverage and interval width on the test set. Summaries of the results are reported in the following Table.

\begin{table}[ht]
\centering
\begin{tabular}{r|rrrrrr}
  \hline
    \multicolumn{7}{c}{Coverage}\\
Method & Min. & 1st Qu. & Median & Mean & 3rd Qu. & Max. \\ 
  \hline
Single & 0.74 & 0.88 & 0.90 & 0.90 & 0.92 & 0.98 \\ 
  Leftskewed & 0.82 & 0.89 & 0.91 & 0.91 & 0.92 & 0.96 \\ 
  Jackknife & 0.88 & 0.94 & 0.95 & 0.95 & 0.96 & 0.98 \\ 
  $\tau = \alpha$ & 0.86 & 0.93 & 0.94 & 0.94 & 0.95 & 0.98 \\ 
 $\tau = 1/2$ & 0.88 & 0.94 & 0.96 & 0.95 & 0.96 & 0.99 \\ 
  $\tau = 1-\alpha$ & 0.94 & 0.97 & 0.98 & 0.98 & 0.99 & 1.00 \\ 
   \hline
      \multicolumn{7}{c}{Interval width}\\
Method & Min. & 1st Qu. & Median & Mean & 3rd Qu. & Max. \\ 
  \hline
Single & 0.28 & 0.44 & 0.49 & 0.51 & 0.55 & 1.12 \\ 
  Leftskewed & 0.00 & 0.46 & 0.49 & 0.50 & 0.53 & 0.68 \\ 
  Jackknife & 0.51 & 0.69 & 0.73 & 0.74 & 0.78 & 1.06 \\ 
  $\tau = \alpha$ & 0.00 & 0.61 & 0.65 & 0.66 & 0.70 & 1.30 \\ 
  $\tau = 1/2$ & 0.00 & 0.66 & 0.72 & 0.72 & 0.78 & 1.02 \\ 
  $\tau = 1-\alpha$ & 0.00 & 0.91 & 0.96 & 0.96 & 1.02 & 1.35 \\ 
   \hline
\end{tabular}
\caption{Summary of coverage and interval width for each method based on 1000 repetitions. }
\end{table}

\end{document}